%% file: article.tex
\documentclass[a4paper,11pt]{article}
\usepackage{a4wide}
\usepackage{amsthm}
\newtheorem{theorem}{Theorem}
\newtheorem{corollary}{Corollary}
\newtheorem{lemma}{Lemma}
\usepackage[utf8]{inputenc}

\usepackage[noadjust]{cite}
\usepackage{amsmath,amssymb,nicefrac,graphicx}
\usepackage{tikz}
\usetikzlibrary{matrix}
\usetikzlibrary{positioning}
\usetikzlibrary{decorations.pathreplacing}

\usepackage{hyperref}
\usepackage{comment}

\RequirePackage{xcolor}
\usepackage{soul} 
\usepackage{framed} 
\definecolor{shadecolor}{rgb}{1,0.8,0.2} 

\usepackage[inline,nomargin,author=,draft]{fixme}
\fxusetheme{color}

\usepackage{subfig}
\newcommand{\stringset}{\mathcal{S}}

\newcommand{\sacc}{\mathsf{access}}
\newcommand{\srep}{\mathsf{replace}}
\newcommand{\sins}{\mathsf{insert}}
\newcommand{\sdel}{\mathsf{delete}}
\newcommand{\ssplit}{\mathsf{split}}
\newcommand{\sconcat}{\mathsf{concat}}

\newcommand{\psum}{\mathsf{sum}}
\newcommand{\psearch}{\mathsf{search}}
\newcommand{\pupdate}{\mathsf{update}}

\newcommand{\pentrydivide}{\mathsf{divide}}
\newcommand{\pentrymerge}{\mathsf{merge}}
\newcommand{\pinsert}{\mathsf{insert}}
\newcommand{\pdelete}{\mathsf{delete}}

\newcommand{\pdivide}{\mathsf{divide}}
\newcommand{\pmerge}{\mathsf{merge}}

\newcommand{\updbit}{\delta}
\newcommand{\updabs}{\Delta}

\newcommand{\rank}{\mathrm{rank}}

\renewcommand{\succ}{\ensuremath{\mathsf{succ}}}
\newcommand{\pred}{\ensuremath{\mathsf{pred}}}
\newcommand{\trep}{\ensuremath{\mathsf{rep}}}
\newcommand{\trank}{\ensuremath{\mathsf{rank}}}
\newcommand{\tselect}{\ensuremath{\mathsf{select}}}
\newcommand{\tmember}{\ensuremath{\mathsf{member}}}

\makeatletter
\def\@fnsymbol#1{\ensuremath{\ifcase#1\or \star\or \dagger\or
   \mathsection\or \mathparagraph\or \ddagger\or \|\or \star\star\or \dagger\dagger
   \or \ddagger\ddagger \else\@ctrerr\fi}}
\renewcommand{\thefootnote}{\fnsymbol{footnote}}
\makeatother

\title{Dynamic Relative Compression, Dynamic Partial Sums, and Substring Concatenation\footnote{An extended abstract appeared in the proceedings of the 27th International Symposium on Algorithms and Computation (ISAAC).}}

\author{
    Philip Bille   \and Patrick Hagge Cording  \and Inge Li G{\o}rtz  
    \and Frederik Rye Skjoldjensen\and Hjalte Wedel Vildh{\o}j\and S{\o}ren Vind
 }

\date{\small Technical University of Denmark
}

\begin{document}
\pagestyle{plain}	

\maketitle
\renewcommand{\thefootnote}{\arabic{footnote}}
\setcounter{footnote}{0}

\begin{abstract}
\noindent Given a static reference string $R$ and a source string $S$, a relative compression of $S$ with respect to $R$ is an encoding of $S$ as a sequence of references to substrings of $R$. Relative compression schemes are a classic model of compression and have recently proved very successful for compressing highly-repetitive massive data sets such as genomes and web-data. We initiate the study of relative compression in a dynamic setting where the compressed source string $S$ is subject to edit operations. The goal is to maintain the compressed representation compactly, while supporting edits and allowing efficient random access to the (uncompressed) source string. We present new data structures that achieve optimal time for updates and queries while using space linear in the size of the optimal relative compression, for nearly all combinations of parameters. We also present solutions for restricted and extended sets of updates. To achieve these results, we revisit the dynamic partial sums problem and the substring concatenation problem. We present new optimal or near optimal bounds for these problems. Plugging in our new results we also immediately obtain new bounds for the string indexing for patterns with wildcards problem and the dynamic text and static pattern matching problem. 
\end{abstract}

\section{Introduction}
Given a static reference string $R$ and a source string $S$, a \emph{relative compression of $S$ with respect to $R$} is an encoding of $S$ as a sequence of references to substrings of $R$. Relative compression (or \emph{external macro compression}) is a classic model of compression defined by Storer~and~Szymanski~\cite{SS1978,Storer1982} in 1978 and has since been used in a wide range of compression scenarios~\cite{DBLP:journals/todaes/LiaoDK99, DBLP:journals/dafes/LiaoDKTW98, kuruppu2010relative, kuruppu2011optimized, chern2012reference, do2014fast, hoobin2011relative}. To compress massive highly-repetitive data sets, such as biological sequences and web collections, relative compression has been shown to be very practical~\cite{kuruppu2010relative, kuruppu2011optimized, hoobin2011relative}.

Relative compression is often applied to compress multiple similar source strings. In such settings relative compression is superior to compressing the source strings individually. For instance, human genomes are 99\% similar and hence relative compression might be used to compress a large collection of sequenced genomes using, e.g., the human reference genome as the static reference string. We focus on the case of compressing a single source string, but our results trivially generalize to compressing multiple source strings.

In this paper we initiate the study of relative compression in a \emph{dynamic setting}, where the compressed source string $S$ is subject to edit operations (insertions, deletions, and replacements of single characters). The goal is to maintain the compressed representation compactly, while supporting edits and allowing efficient random access to the (uncompressed) source string. Efficient data structures supporting these operations allow us to avoid costly recompression of massive data sets after updates. 

We provide the first non-trivial bounds for this problem. We present new data structures that achieve \emph{optimal} time for updates and queries while using space linear in the size of the \emph{optimal} relative compression, for nearly all combinations of parameters. We also present solutions for restricted and extended sets of updates.

To achieve these results, we revisit the \emph{dynamic partial sums problem} and the \emph{substring concatenation problem}. We present new optimal or near optimal bounds for both of these problems (see detailed discussion below). Furthermore, plugging in our new results immediately leads to new bounds for the \emph{string indexing for patterns with wildcards problem}~\cite{LNV2014,BGVV2014} and the \emph{the dynamic text and static pattern matching problem}~\cite{amir2007dynamic}.

\subsection{Dynamic Relative Compression}
Given a \emph{reference string} $R$ and a \emph{source string} $S$,  a \emph{relative compression of $S$ with respect to $R$} is a sequence $C=(i_1,j_1),...,(i_{|C|},j_{|C|})$ such that $S = R[i_1,j_1] \cdots R[i_{|C|},j_{|C|}]$. We call $C$ a \emph{substring cover} for $S$. The substring cover is \emph{optimal} if $|C|$ is minimum over all relative compressions of $S$ with respect to $R$. The \emph{dynamic relative compression problem} is to maintain a relative compression of $S$ under the following operations. Let $i$ be a position in $S$ and $\alpha$ be a character.

\begin{description}
\item[$\quad\sacc(i)$:] return the character $S[i]$,
\item[$\quad\srep(i, \alpha)$:] change $S[i]$ to character $\alpha$,
\item[$\quad\sins(i, \alpha)$:] insert character $\alpha$ before position $i$ in $S$,
\item[$\quad\sdel(i)$: ] delete the character at position $i$ in $S$.
\end{description}

\noindent Note that operations $\sins$ and $\sdel$ change the length of $S$ by a single character. In all bounds below, the $\sacc(i)$ operation extends to decompressing an arbitrary substring of length $\ell$ using only $O(\ell)$ additional time.

\paragraph{Our Results}
Throughout the paper, let $r$ be the length of the reference string $R$, $N$ be the length of the (uncompressed) string $S$, and $n$ be the size of an optimal relative compression of $S$ with regards to $R$. All of the bounds mentioned below and presented in this paper hold for a standard unit-cost RAM with $w$-bit words with standard arithmetic and logical operations on a word. This means that the algorithms can be implemented directly in standard imperative programming languages such as C~\cite{KR78} or C++~\cite{Sto00}. An index into $R$ or $S$ can be stored in a single word and hence $w \geq \log (n+r)$.

\begin{theorem}\label{thm:main}
Let $R$ and $S$ be a reference and source string of lengths $r$ and $N$, respectively, and let $n$ be the length of the optimal substring cover of $S$ by $R$. Then, we can solve the dynamic relative compression problem supporting $\sacc$, $\srep$, $\sins$, and $\sdel$
\begin{itemize}
\item[(i)] in $O(n + r)$ space and $O\left(\frac{\log n}{\log \log n} + \log \log r\right)$ time per operation, or
\item[(ii)] in $O(n + r \log^\epsilon r)$ space and $O\left(\frac{\log n}{\log \log n}\right)$ time per operation, for any constant $\epsilon>0$.
\end{itemize}
\end{theorem}

\noindent These are the first non-trivial bounds for the problem. Together, the bounds are optimal for most natural parameter combinations. In particular, any data structure for a string of length $N$ supporting $\sacc$, $\sins$, and $\sdel$ must use $\Omega(\log N/\log \log N)$ time in the worst-case regardless of the space~\cite{fredman1989cell} (this is called the \emph{list representation problem}). Since $n \leq N$, we can view  $O(\log n /\log \log n)$ as a compressed version of the optimal time bound that is always $O(\log N/\log \log N)$ and better when $S$ is compressible. Hence, Theorem~\ref{thm:main}(i) provides a linear-space solution that achieves the compressed time bound except for an $O(\log \log r)$ additive term. Note that whenever $n \geq (\log r)^{\log^\epsilon \log r}$, for any $\epsilon>0$, the $\log n/\log \log n$ term dominates the query time and we match the compressed time bound. Hence, Theorem~\ref{thm:main}(i) is only suboptimal in the special case when $n$ is almost exponentially smaller than $r$. In this case, we can use Theorem~\ref{thm:main}(ii) which always provides a solution achieving the compressed time bound at the cost of increasing the space to $O(n + r\log^\epsilon r)$.

We note that dynamic compression under different models of compression has been studied extensively~\cite{grossi03, Ferragina04succinctrepresentation, ferragina2005indexing, Sadakane2006, Ferragina2007, Jansson12_cram, navarro2013optimal}. However, all of these results require space dependent on the size of the original string and hence cannot take full advantage of highly-repetitive data.

\subsection{Dynamic Partial Sums} The \emph{partial sums problem} is to maintain an array $Z[1..s]$ under the following operations.

\begin{description}
 \item[$\quad\psum(i)$:] return $\sum_{j=1}^i Z[j]$,
 \item[$\quad\pupdate(i, \updabs)$:] set $Z[i] = Z[i] + \updabs$,
 \item[$\quad\psearch(t)$:] return $1 \leq i \leq s$ such that $\psum(i-1) < t \leq \psum(i)$. To ensure well-defined answers, we require that $Z[i] \geq 0$ for all $i$. 
\end{description}

\noindent The partial sums problem is a classic and well-studied problem~\cite{dietz1989optimal,raman2001succinct,husfeldt2003new,fredman1989cell,hon2011a,husfeldt1996lower,fenwick1994new,puaatracscu2004tight}. In our context, we consider the problem in the word RAM model, where each array entry stores a $w$-bit integer and the element of the array can be changed by $\updbit$-bit integers, i.e., the argument $\updabs$ can be stored in $\updbit$ bits. In this setting, Pătraşcu and Demaine~\cite{puaatracscu2004tight} gave a linear-space data structure with $\Theta(\log s / \log (w / \updbit))$ time per operation. They also gave a matching lower bound.

We consider the following generalization supporting dynamic changes to the array. The \emph{dynamic partial sums problems} is to additionally support the following operations.

\begin{description}
 \item[$\quad\pinsert(i, \updabs)$:] insert a new entry in $Z$ with value $\updabs$ before $Z[i]$,
 \item[$\quad\pdelete(i)$:] delete the entry $Z[i]$ of value at most $\updabs$.
 \item[$\quad\pentrymerge(i)$:] replace entry $Z[i]$ and $Z[i+1]$ with a new entry with value $Z[i] + Z[i+1]$.
 \item[$\quad\pentrydivide(i, t)$:], where $0 \leq t \leq Z[i]$. Replace entry $Z[i]$ by two new consecutive entries with value $t$ and $Z[i] - t$, respectively. 
\end{description}

\noindent Hon~et~al.~\cite{hon2011a} and Navarro and Sadakane~\cite{navarro2014fully} presented optimal solutions for this problem in the case where the entries in $Z$ are at most polylogarithmic in $s$ (they did not explicitly consider the $\pentrymerge$ and $\pentrydivide$ operation). 

\paragraph{Our Results}
We show the following improved result. 

\begin{theorem}\label{thm:partialsums}
 Given an array of length $s$ storing $w$-bit integers and parameter $\updbit$, such that $\updabs < 2^{\updbit}$,  we can solve the dynamic partial sums problem supporting $\psum$, $\pupdate$, $\psearch$, $\pinsert$, $\pdelete$, $\pentrymerge$, and $\pentrydivide$ in linear space and $O(\log s / \log (w / \updbit))$ time per operation.
 \end{theorem}
\noindent Note that this bound simultaneously matches the optimal time bound for the standard partial sums problem and supports storing arbitrary $w$-bit values in the entries of the array, i.e., the values we can handle in optimal time are exponentially larger than in the previous results.

To achieve our bounds we extend the static solution by Pătraşcu and Demaine~\cite{puaatracscu2004tight}. Their solution is based on storing a sampled subset of \emph{representative elements} of the array and difference encode the remaining elements. They pack multiple difference encoded elements in words and then apply word-level parallelism to speedup the operations. To support $\pinsert$ and $\pdelete$ the main challenge is to  maintain the representative elements that now dynamically move within the array. We show how to efficiently do this by combining a new representation of representative elements with a recent result by Pătraşcu and Thorup~\cite{patrascu2014dynamic}. Along the way we also slightly simplify the original construction by Pătraşcu and Demaine~\cite{puaatracscu2004tight}.

\subsection{Substring Concatenation}
Let $R$ be a string of length $r$. A \emph{substring concatenation query} on $R$ takes two pairs of indices $(i,j)$ and $(i',j')$ and returns the start position in $R$ of an occurrence of $R[i,j] R[i',j']$, or \texttt{NO} if the string is not a substring of $R$. The \emph{substring concatenation problem} is to preprocess $R$ into a data structure that supports substring concatenation queries.

Amir~et~al.~\cite{amir2007dynamic} gave a solution using $O(r\sqrt{\log r})$ space with query time $O(\log\log r)$, and recently Gawrychowski~et~al.~\cite{gawrychowski2014weighted} showed how to solve the problem in $O(r\log r)$ space and $O(1)$ time. 

\paragraph{Our Results} 
We give the following improved bounds.
\begin{theorem}\label{thm:substringconcat}
Given a string $R$ of length $r$, the substring concatenation problem can be solved in either
\begin{itemize}
\item[(i)] $O(r\log^\epsilon r)$ space and $O(1)$ time, for any constant $\epsilon>0$, or
\item[(ii)] $O(r)$ space and $O(\log \log r)$ time.
\end{itemize}
\end{theorem}

\noindent Hence, Theorem~\ref{thm:substringconcat}(i) matches the previous $O(1)$ time bound while reducing the space from $O(r\log r)$ to $O(r\log^\epsilon r)$ and Theorem~\ref{thm:substringconcat}(ii) achieves linear space while using $O(\log \log r)$ time. Plugging in the two solutions into our solution for dynamic relative compression leads  to the two branches of Theorem~\ref{thm:main}. 

To achieve the bound in (i), the main idea is a new construction that efficiently combines compact data structure for 1D range reporting~\cite{belazzougui2010fast} with the recent constant time weighted level ancestor data structure for suffix trees~\cite{gawrychowski2014weighted}. The bound in (ii) follows as a simple implication of another recent result for \emph{unrooted LCP queries}~\cite{BGVV2014} by some of the authors. The substring concatenation problem is a key component in several solutions to the \emph{string indexing for patterns with wildcards problem}~\cite{BGVV2014, cole2004dictionary, LNV2014}, where the goal is to preprocess a string $T$ to support pattern matching queries for patterns with wildcards. Plugging in Theorem~\ref{thm:substringconcat}(i) we immediately obtain the following new bound for the problem. 

\begin{corollary}
Let $T$ be a string of length $t$. For any pattern string $P$ of length $p$ with $k$ wildcards, we can support pattern matching queries on $T$ using $O(t\log^\epsilon t)$ space and $O(p + \sigma^k)$ time for any constant $\epsilon>0$.
\end{corollary}

\noindent This improves the running time of fastest linear space solution by a factor $\log \log t$ at the cost of increasing the space slightly by a factor $\log^\epsilon t$. See \cite{LNV2014} for detailed overview of the known results. 

\subsection{Extensions}
Finally, we present two extensions of the dynamic relative compression problem. 

\subsubsection{Dynamic Relative Compression with Access and Replace}
If we restrict the operations to $\sacc$ and $\srep$ we obtain the following improved bound. 

\begin{theorem}\label{thm:extension1}
Let $R$ and $S$ be a reference and source string of lengths $r$ and $N$, respectively, and let $n$ be the length of the optimal substring cover of $S$ by $R$. Then, we can solve the dynamic relative compression problem supporting $\sacc$ and $\srep$ 
in $O(n + r)$ space and $O(\log\log N)$ expected time.
\end{theorem}

\noindent This version of dynamic relative compression is a key component in the \emph{dynamic text and static pattern matching problem}, where the goal is to efficiently maintain a set of occurrences of a pattern $P$ in a text $T$ that is dynamically updated by changing individual characters. Let $p$ and $t$ denote the lengths of $P$ and $T$, respectively. Amir~et~al.~\cite{amir2007dynamic} gave a data structure using $O(t+p\sqrt{\log p})$ space which supports updates in $O(\log \log p)$ time. The computational bottleneck in the update operation is to update a substring cover of size $O(p)$. Plugging in the bounds from Theorem~\ref{thm:extension1}, we immediately obtain the following improved bound.

\begin{corollary}
Given a pattern $P$ and text $T$ of lengths $p$ and $t$, respectively, we can solve the dynamic text and static pattern matching problem in $O(t + p)$ space and $O(\log\log p)$ expected time per update.
\end{corollary}

\noindent Hence, we match the previous time bound while improving the space to linear. 

\subsubsection{Dynamic Relative Compression with Split and Concatenate}
We also consider maintaining a set of compressed strings under split and concatenate operations (as in Alstrup et al.~\cite{ABR2000}). Let $R$ be a reference string and let $\stringset=\{S_1,\ldots,S_k\}$ be a set of strings compressed relative to $R$. In addition to $\sacc$, $\srep$, $\sins$ and $\sdel$ we also define the following operations. 
\begin{description}
\item[$\quad\sconcat(i,j)$:] Add string $S_i \cdot S_j$ to $\stringset$ and remove $S_i$ and $S_j$. 
\item[$\quad\ssplit(i,j)$:] Remove $S_i$ from $\stringset$ and add $S_i[1,j-1]$ and $S_i[j,|S_i|]$.
\end{description}

\noindent We obtain the following bounds. 
\begin{theorem}\label{thm:extension2}
Let $R$ be a reference string of length $r$, let $\stringset=\{S_1,\ldots,S_k\}$ be a set of source strings of total length $N$, and let $n$ be the total length of the optimal substring covers of the strings in $\stringset$. Then, we can solve the dynamic relative compression problem supporting $\sacc$, $\srep$, $\sins$, $\sdel$, $\ssplit$, and $\sconcat$,
\begin{itemize}
\item[(i)] in space $O(n + r)$ and time $O(\log n)$ for $\sacc$ and time $O(\log n + \log\log r)$ for $\srep$, $\sins$, $\sdel$, $\ssplit$, and $\sconcat$, or
\item[(ii)] in space $O(n + r\log^\epsilon r)$ and time $O(\log n)$ for all operations.
\end{itemize}
\end{theorem}

\noindent Hence, compared to the bounds in Theorem~\ref{thm:main} we only increase the time bounds by an additional $\log \log n$ factor.

\section{Dynamic Relative Compression}\label{sec:DRC}
In this section we show how Theorems~\ref{thm:partialsums} and \ref{thm:substringconcat} lead to Theorem~\ref{thm:main}. The proofs of Theorems~\ref{thm:partialsums} and \ref{thm:substringconcat} appear in Section~\ref{sec:dynps} and Section~\ref{app:substringconcat}, respectively. 

Let $C=((i_1,j_1),...,(i_{|C|},j_{|C|}))$ be the compressed representation of $S$. From now on, we refer to $C$ as the \emph{cover of $S$}, and call each element $(i_l,j_l)$ in $C$ a \emph{block}. Recall that a block $(i_l,j_l)$ refers to a substring $R[i_l,j_l]$ of $R$.
A cover $C$ is \emph{maximal} if 
concatenating any two consecutive blocks $(i_l,j_l),(i_{l+1},j_{l+1})$ in $C$ yields a string that does not occur in $R$, 
i.e., the string $R[i_l,j_l]R[i_{l+1},j_{l+1}]$ is not a substring of $R$. We need the following lemma.

\begin{lemma}\label{lem:maxcover}
If $C_\textsc{max}$ is a maximal cover and $C$ is an arbitrary cover of $S$, then $|C_\textsc{max}| \leq 2|C|-1$.
\end{lemma}
\begin{proof}
  In each block $b$ of $C$ there can start at most two blocks in $C_\textsc{max}$, because otherwise two adjacent blocks in $C_\textsc{max}$ would be entirely contained in the block $b$, contradicting the maximality of $C_\textsc{max}$. Since the last block of both $C$ and $C_\textsc{max}$ end at the last position of $S$, a contradiction of the maximality is already obtained when more than one block of $C_\textsc{max}$ start in the last block of $C$. Hence, $|C_\textsc{max}| \leq 2|C|-1$.
\end{proof}

\noindent Recall that $n$ is the size of an optimal cover of $S$ with regards to $R$. 
The lemma implies that we can maintain a compression of size at most $2n-1$ by maintaining a maximal cover of $S$.
The remainder of this section describes our data structure for maintaining and accessing such a cover.

Initially, we can use the suffix tree of $R$ to construct a maximal cover of $S$ in $O(N+r)$ time by greedily matching the maximal prefix of the remaining part of $S$ with any suffix of $R$. This guarantees that the blocks constitute a maximal cover of $S$.

\subsection{Data Structure}
The high level idea for supporting the operations on $S$ is to store the sequence of block lengths $j_1-i_1+1, \ldots, j_{|C|}-i_{|C|}+1$ in a dynamic partial sums data structure. This allows us, for example, to identify the block that encodes the $k^\text{th}$ character in $S$ by performing a $\psearch(k)$ query.

Updates to $S$ are implemented by splitting a block in $C$. This may break the maximality property so we use substring concatenation queries on $R$ to detect if blocks can be merged. We only need a constant number of substring concatenation queries to restore maximality. To maintain the correct sequence of block lengths we use $\pupdate$, $\pentrydivide$ and $\pentrymerge$ operations on the dynamic partial sums data structure.

Our data structure consist of the string $R$, a substring concatenation data structure of Theorem~\ref{thm:substringconcat} for $R$, a maximal cover $C$ for $S$ stored in a doubly linked list, and the dynamic partial sums data structure of Theorem~\ref{thm:partialsums} storing the block lengths of $C$. We also store auxiliary links between a block in the doubly linked list and the corresponding block length in the partial sums data structure, and a list of alphabet symbols in $R$ with the location of an occurrence for each symbol.
By Lemma~\ref{lem:maxcover} and since $C$ is maximal we have $|C| \leq 2n-1 = O(n)$. Hence, the total space for $C$ and the partial sums data structure is $O(n)$. The space for $R$ is $O(r)$ and the space for substring concatenation data structure is either $O(r)$ or $O(r \log^\epsilon r)$ depending on the choice in Lemma~\ref{thm:substringconcat}. Hence, in total we use either $O(n + r)$ or $O(n +r \log^\epsilon r)$ space.

\subsection{Answering Queries}
To answer $\sacc(i)$ queries we first compute $\psearch(i)$ in the dynamic partial sums structure to identify the block $b_l = (i_l,j_l)$ containing position $i$ in $S$. The local index in $R[i_l,j_l]$ of the $i^\text{th}$ character in $R$ is $\ell = i - \psum(l-1)$, and thus the answer to the query is the character $R[i_l + \ell -1]$.

We perform $\srep$ and $\sdel$ by first identifying $b_l = (i_l,j_l)$ and $\ell$ as above. Then we partition $b_l$ into three new blocks 
${b^1_l}=(i_l, i_l+\ell-2)$, ${b^2_l} = (i_l+\ell-1, i_l+\ell-1)$, ${b^3_l}=(i_l+\ell, j_l)$ where $b^2_l$ is the single character block for index $i$ in $S$ that we must change. 
In $\srep$ we change ${b^2_l}$ to an index of an occurrence in $R$ of the new character (which we can find from the list of alphabet symbols), while we remove ${b^2_l}$ in $\sdel$. The new blocks and their neighbors, that is, $b_{l-1}$, ${b^1_l}$, ${b^2_l}$, ${b^3_l}$, and $b_{l+1}$ may now be non-maximal. To restore maximality we perform substring concatenation queries on each consecutive pair of these 5 blocks, and replace non-maximal blocks with merged maximal blocks.
All other blocks are still maximal, since the strings obtained by concatenating $b_{l'}$ with $b_{l'+1}$, for all $l'<l-1$ and all $l'>l$, was not present in $R$ before the change and is not present afterwards.
A similar idea is used by Amir~et~al.~\cite{amir2007dynamic}. We perform $\pupdate$, $\pentrydivide$ and $\pentrymerge$ operations to maintain the corresponding lengths in the dynamic partial sums data structure. The $\sins$ operation is similar, but inserts a new single character block between two parts of $b_l$ before restoring maximality. Observe that using $\updbit = O(1)$ bits in $\pupdate$ is sufficient to maintain the correct block lengths.

In total, each operation requires a constant number of substring concatenation queries and dynamic partial sums operations; the latter having time complexity $O(\log n/\log(w/\updbit))=O(\log n/\log\log n)$ as $w \geq \log n$ and $\updbit = O(1)$. Hence, the total time for each $\sacc$, $\srep$, $\sins$, and $\sdel$ operation is either $O(\log n/\log \log n + \log \log r)$ or $O(\log n/\log \log n)$ depending on the substring concatenation data structure used. In summary, this proves Theorem~\ref{thm:main}.

\section{Dynamic Partial Sums}\label{sec:dynps}

In this section we prove Theorem~\ref{thm:partialsums}. We support the operations $\pinsert(i, \updabs)$ and $\pdelete(i)$ on a sequence of $w$-bit integer keys by implementing them using $\pupdate$ and a $\pentrydivide$ or $\pentrymerge$ operation, respectively. This means that we support inserting or deleting keys with value at most $2^\updbit$.

We first solve the problem for small sequences. The general solution uses a standard reduction, storing $Z$ at the leaves of a B-tree of large outdegree. We use the solution for small sequences to navigate in the internal nodes of the B-tree.

\paragraph{Dynamic Integer Sets}
We need the following recent result due to Pătraşcu and Thorup~\cite{patrascu2014dynamic} on maintaining a set of integer keys $X$ under insertions and deletions.
The queries are as follows, where $q$ is an integer. The membership query $\tmember(q)$ returns true if $q \in X$, predecessor $\pred_X(q)$ returns the largest key $x \in X$ where $x < q$, and successor $\succ_X(q)$ returns the smallest key $x \in X$ where $x \geq q$. The rank $\trank_X(q)$ returns the number of keys in $X$ smaller than $q$, and $\tselect(i)$ returns the $i^\text{th}$ smallest key in $X$.

\begin{lemma}[Pătraşcu and Thorup~\cite{patrascu2014dynamic}]\label{lem:dynIS}
    There is a data structure for maintaining a dynamic set of $w^{O(1)}$ $w$-bit integers that supports insert, delete, membership, predecessor, successor, rank and select in constant time per operation.
    
\end{lemma}

\subsection{Dynamic Partial Sums for Small Sequences}
Let $Z$ be a sequence of at most $B \leq w^{O(1)}$ integer keys. We will show how to store $Z$ in linear space such that all dynamic partial sums operations can be performed in constant time. We let $Y$ be the sequence of prefix sums of $Z$, defined such that each key $Y[i]$ is the sum of the first $i$ keys in $Z$, i.e., $Y[i] = \sum_{j=1}^i Z[j]$. 
Observe that $\psum(i) = Y[i]$ and $\psearch(t)$ is the index of the successor of $t$ in $Y$. Our goal is to store and maintain a representation of $Y$ subject to the dynamic operations $\pupdate$, $\pentrydivide$ and $\pentrymerge$ in constant time per operation.

\subsubsection{The Scheme by Pătraşcu and Demaine}
We first review the solution to the static partial sums problem by Pătraşcu and Demaine~\cite{puaatracscu2004tight}, slightly simplified due to Lemma~\ref{lem:dynIS}. Our dynamic solution builds on this.

The entire data structure is rebuilt every $B$ operations as follows. We first partition $Y$ greedily into \emph{runs}. Two adjacent elements in $Y$ are in the same run if their difference is at most $B 2^\updbit$, and we call the first element of each run a \emph{representative} for all elements in the run. We use $\mathcal{R}$ to denote the sequence of representative values in $Y$ and $\trep(i)$ to be the index of the representative for element $Y[i]$ among the elements in $\mathcal{R}$.

We store $Y$ by splitting representatives and other elements into separate data structures: $\mathcal{I}$ and $\mathcal{R}$ store the representatives at the time of the last rebuild, while $\mathcal{U}$ stores each element in $Y$ as an offset to its representative value as well as updates since the last rebuild. We ensure $Y[i] = \mathcal{R}[\trep(i)]+\mathcal{U}[i]$ for any $i$ and can thus reconstruct the values of $Y$.

The representatives are stored as follows. $\mathcal{I}$ is the sequence of indices in $Y$ of the representatives and $\mathcal{R}$ is the sequence of representative values in $Y$. Both $\mathcal{I}$ and $\mathcal{R}$ are stored using the data structure of Lemma~\ref{lem:dynIS}. We can then define $\trep(i) = \trank_\mathcal{I}(\pred_\mathcal{I}(i))$ as the index of the representative for $i$ among all representatives, and use $\mathcal{R}[\trep(i)] = \tselect_\mathcal{R}(\trep(i))$ to get the value of the representative for $i$.

We store in $\mathcal{U}$ the current difference from each element to its representative, $\mathcal{U}[i] = Y[i] - \mathcal{R}[\trep(i)]$ (i.e. updates between rebuilds are applied to $\mathcal{U}$). The idea is to pack $\mathcal{U}$ into a single word of $B$ elements.
Observe that $\pupdate(i, \updabs)$ adds value $\updabs$ to all elements in $Y$ with index at least $i$. We can support this operation in constant time by adding to $\mathcal{U}$ a word that encodes $\updabs$ for those elements.
Since each difference between adjacent elements in a run is at most $B 2^\updbit$ and $|Y| = O(B)$, the maximum value in $\mathcal{U}$ after a rebuild is $O(B^2 2^\updbit)$. As $B$ updates of size $2^\updbit$ may be applied before a rebuild, the changed value at each element due to updates is $O(B 2^\updbit)$. So each element in $\mathcal{U}$ requires $O(\log B + \updbit)$ bits (including an overflow bit per element). Thus, $\mathcal{U}$ requires $O(B(\log B + \updbit))$ bits in total and can be packed in a single word for $B = O(\min \{w / \log w, w / \updbit \})$.

Between rebuilds the stored representatives are potentially outdated because updates may have changed their values. 
However, observe that the values of two consecutive representatives differ by more than $B 2^\updbit$ at the time of a rebuild, so the gap between two representatives cannot be closed by $B$ updates of $\updbit$ bits each (before the structure is rebuilt again). 
Hence, an answer to $\psearch(t)$ cannot drift much from the values stored by the representatives; it can only be in a constant number of runs, namely those with a representative value $\succ_\mathcal{R}(t)$ and its two neighboring runs.
In a run with representative value $v$, we find the smallest $j$ (inside the run) such that $\mathcal{U}[j] + v - t > 0$. The smallest $j$ found in all three runs is the answer to the $\psearch(t)$ query.
Thus, by rebuilding periodically, we only need to check a constant number of runs when answering a $\psearch(t)$ query.

On this structure, Pătraşcu and Demaine~\cite{puaatracscu2004tight} show that the operations $\psum$, $\psearch$ and $\pupdate$ can be supported in constant time each as follows:
\begin{description}
\item[$\psum(i)$:] return the sum of $\mathcal{R}[\trep(i)]$ and $\mathcal{U}[i]$. This takes constant time as $\mathcal{U}[i]$ is a field in a word and representatives are stored using Lemma~\ref{lem:dynIS}.
\item[$\psearch(t)$:] let $r_0 = \trank_\mathcal{R}(\succ_\mathcal{R}(t))$. We must find the smallest $j$ such that $\mathcal{U}[j] + R[r] - t > 0$ for $r \in \{r_0-1, r_0, r_0+1\}$, where $j$ is in run $r$.
We do this for each $r$ using standard word operations in constant time by adding $R[r]-t$ to all elements in $\mathcal{U}$, masking elements not in the run (outside indices $\tselect_\mathcal{I}(r)$ to $\tselect_\mathcal{I}(r+1)-1$, and counting the number of negative elements.
\item[$\pupdate(i, \updabs)$:] we do this in constant time by copying $\updabs$ to all fields $j \geq i$ by a multiplication and adding the result to $\mathcal{U}$. 
\end{description}
To count the number of negative elements or find the least significant bit in a word in constant time, we use the technique by Fredman and Willard~\cite{fredmanwillardfusion}.

Notice that rebuilding the data structure every $B$ operations takes $O(B)$ time, resulting in amortized constant time per operation. We can instead do this incrementally by a standard approach by Dietz~\cite{dietz1989optimal}, reducing the time per operation to worst case constant. 
The idea is to construct the new replacement data structure incrementally while using the old and complete data structure. 

\subsubsection{Efficient Support for $\pentrydivide$ and $\pentrymerge$}
We now show how to maintain the structure described above while supporting operations $\pentrydivide(i, t)$ and $\pentrymerge(i)$. An example supporting the following explanation is provided in Figure \ref{fig:exampleinit}.

\input{arrayexample1}

Observe that the operations are only local: Splitting $Z[i]$ into two parts or merging $Z[i]$ and $Z[i+1]$ does not influence the precomputed values in $Y$ (besides adding/removing values for the divided/merged elements). We must update $\mathcal{I}$, $\mathcal{R}$ and $\mathcal{U}$ to reflect these local changes accordingly. Because a $\pentrydivide$ or $\pentrymerge$ operation may create new representatives between rebuilds with values that do not fit in $\mathcal{U}$, we change $\mathcal{I}$, $\mathcal{R}$ and $\mathcal{U}$ to reflect these new representatives by rebuilding the data structure locally. This is done as follows.

Consider the run representatives. Both $\pentrydivide(i, t)$ and $\pentrymerge(i)$ may require us to create a new run, combine two existing runs or remove a run.
In any case, we can find a replacement representative for each run affected. 
As the operations are only local, the replacement is either a divided or merged element, or one of the neighbors of the replaced representative. 
Replacing representatives may cause both indices and values for the stored representatives to change.
We use insertions and deletions on $\mathcal{R}$ to update representative values. 

Since the new operations change the indices of the elements, these changes must also be reflected in $\mathcal{I}$. For example, a $\pentrymerge(i)$ operation decrements the indices of all elements with index larger than $i$ compared to the indices stored at the time of the last rebuild
We should in principle adjust the $O(B)$ changed indices stored in $\mathcal{I}$. The cost of adjusting the indices accordingly when using Lemma~\ref{lem:dynIS} to store $\mathcal{I}$ is $O(B)$. 
Instead, to get our desired constant time bounds, we represent $\mathcal{I}$ using a resizable data structure with the same number of elements as $Y$ that supports this kind of update. We must support $\tselect_\mathcal{I}(i)$, $\trank_\mathcal{I}(q)$, and $\pred_\mathcal{I}(q)$ as well as inserting and deleting elements in constant time. Because $\mathcal{I}$ has few and small elements, we can support the operations in constant time by representing it using a bitstring $\mathcal{B}$ and a structure $\mathcal{C}$ which is the prefix sum over $\mathcal{B}$ as follows.

Let $\mathcal{B}$ be a bitstring of length $|Y|\leq B$, where $\mathcal{B}[i] = 1$ iff there is a representative at index $i$. $\mathcal{C}$ has $|Y|$ elements, where $\mathcal{C}[i]$ is the prefix sum of $\mathcal{B}$ including element $i$. Since $\mathcal{C}$ requires $O(B \log B)$ bits in total we can pack it in a single word. We answer queries as follows: $\trank_\mathcal{I}(q)$ equals $\mathcal{C}[q-1]$, we answer $\tselect_\mathcal{I}(i)$ by subtracting $i$ from all elements in $\mathcal{C}$ and return one plus the number of elements smaller than $0$ (as done in $\mathcal{U}$ when answering $\psearch$), and we find $\pred_\mathcal{I}(q)$ as the index of the least significant bit in $\mathcal{B}$ after having masked all indices larger than $q$.
Updates are performed as follows. Using mask, shift and concatenate operations, we can ensure that $\mathcal{B}$ and $\mathcal{C}$ have the same size as $Y$ at all times (we extend and shrink them when performing $\pentrydivide$ and $\pentrymerge$ operations). Inserting or deleting a representative is to set a bit in $\mathcal{B}$, and to keep $\mathcal{C}$ up to date, we employ the same $\pm 1$ update operation as used in $\mathcal{U}$.


We finally need to adjust the relative offsets of all elements with a changed representative in $\mathcal{U}$ (since they now belong to a representative with a different value).
In particular, if the representative for $\mathcal{U}[j]$ changed value from $v$ to $v'$, we must subtract $v' - v$ from $\mathcal{U}[j]$. 
This can be done for all affected elements belonging to a single representative simultaneously in $\mathcal{U}$ by a single addition with an appropriate bitmask ($\pupdate$ a range of $\mathcal{U}$). Note that we know the range of elements to update from the representative indices. Finally, we may need to insert or delete an element in $\mathcal{U}$, which can be done easily by mask, shift and concatenate operations on the word $\mathcal{U}$. This leads to Theorem \ref{thm:prefixsumssmall}.
%


\begin{theorem}\label{thm:prefixsumssmall}
    There is a linear space data structure for dynamic partial sums supporting each operation $\psearch$, $\psum$, $\pupdate$, $\pinsert$, $\pdelete$, $\pentrydivide$, and $\pentrymerge$ on a sequence of length $O(\min \{w / \log w, w / \updbit \})$ in worst-case constant time.
\end{theorem}

\subsection{Dynamic Partial Sums for Large Sequences}
Willard \cite{willard2000examining} (and implicitly Dietz~\cite{dietz1989optimal}) showed that a leaf-oriented B-tree with out-degree $B$ of height $h$ can be maintained in $O(h)$ worst-case time if: 1) searches, insertions and deletions take $O(1)$ time per node when no splits or merges occur, and 2) merging or splitting a node of size $B$ requires $O(B)$ time.

We use this as follows, where $Z$ is our integer sequence of length $s$. Create a leaf-oriented B-tree of degree $B = \Theta(\min \{w / \log w, w / \updbit \})$ storing $Z$ in the leaves, with height $h = O(\log_B n) = O(\log n / \log (w / \updbit))$.
Each node $v$ 
uses Theorem~\ref{thm:prefixsumssmall} to store the $O(B)$ sums of leaves in each of the subtrees of its children. Searching for $t$ in a node corresponds to finding the successor $Y[i]$ of $t$ among these sums. Dividing or merging elements in $Z$ corresponds to inserting or deleting a leaf. This concludes the proof of Theorem~\ref{thm:partialsums}.

\input{substringconcat}

\input{extensions}

\input{conclusion.tex}

\paragraph{Acknowledgments}
We thank Pawel Gawrychowski for helpful discussions.

\bibliographystyle{abbrv}
\bibliography{references}

\end{document}

%% file: arrayexample1.tex

\tikzset{ 
nodesrep/.style={
    rectangle,
    align=center,
    draw=black,
},
table1withindices/.style={
    matrix of nodes,
    row sep=1pt,
    column sep=-\pgflinewidth,
    font=\footnotesize,
    text depth=1pt,
    text height=5pt,
    text width=9pt,
    text centered,
    nodes={
        rectangle,
        draw=black,
        anchor=north,
    },
    nodes in empty cells,
    column 1/.style={
        nodes={
            draw=white,
            align=left
        }
    },
    row 1/.style={
        nodes={
          draw=white,
          font=\tiny
        }
    },
    row 4/.style={
      nodes={
        draw=white
      }
    }
}
}

\begin{figure}
\centering
\begin{tikzpicture}

\matrix (first) [table1withindices]
{
  & 1 & 2 & 3 & 4 & 5 & 6 & 7 & 8 & 9 & 10 & 11 &\node (a) {12}; & 13 & 14 & 15 & 16 & 17 & 18 & 19 \\
  $Z$   & 5 & 1 & 4 & 7 & 1 & 1 & 6 & 5 & 1 & 1 & 2 & 2 & 1 & 3 & 5 & 10 & 5 & 10 & 2 \\
  $Y$   &\node[fill=black!20] () {5}; & 6 & 10 & \node[fill=black!20] () {17}; & 18 & 19 & \node[fill=black!20] () {25}; & \node[fill=black!20] () {30}; & 31 & 32 & 34 & 36 & 37 & 40 & \node[fill=black!20] () {45}; & \node[fill=black!20] () {55}; & \node[fill=black!20] () {60}; & \node[fill=black!20] () {70}; & 72 \\
$\mathcal{R}$  & $\{5, 17, 25, 30, 45, 55, 60, 70\}$ \\
$\mathcal{U}$   & 0 & 1 & 5 & 0 & 1 & 2 & 0 & 0 & 1 & 2 & 4 & 6 & 7 & 10 & 0 & 0 & 0 & 0 & 2 \\
$\mathcal{B}$  & 1 & 0 & 0 & 1 & 0 & 0 & 1 & 1 & 0 & 0 & 0 & 0 & 0 & 0 & 1 & 1 & 1 & 1 & 0 \\
$\mathcal{C}$  & 1 & 1 & 1 & 2 & 2 & 2 & 3 & 4 & 4 & 4 & 4 & 4 & 4 & 4 & 5 & 6 & 7 & 8 & 8 \\
};
\draw (0,-2) node {\textbf{a)} The initial data structure constructed from $Z$.};
\end{tikzpicture}

\vskip 0.5cm

\begin{tikzpicture}

\matrix (first) [table1withindices]
{
  & 1 & 2 & 3 & 4 & 5 & 6 & 7 & 8 & \node (a) {9}; &\node (b) {10}; & 11 &12 & 13 & 14 & 15 & 16 & 17 & 18 & 19 & 20 \\
  $Z$   & 5 & 1 & 4 & 7 & 1 & 1 & 6 & 3 & 2 & 1 & 1 & 2 & 2 & 1 & 3 & 5 & 10 & 5 & 10 & 2 \\
  $Y$   & \node[fill=black!20] () {5}; & 6 & 10 & \node[fill=black!20] () {17}; & 18 & 19 & \node[fill=black!20] () {25}; & 28 & 30 & 31 & 32 & 34 & 36 & 37 & 40 & \node[fill=black!20] () {45}; & \node[fill=black!20] () {55}; & \node[fill=black!20] () {60}; & \node[fill=black!20] () {70}; & 72 \\
$\mathcal{R}$  & $\{5, 17, 25, 45, 55, 60, 70\}$ \\
$\mathcal{U}$   & 0 & 1 & 5 & 0 & 1 & 2 & 0 & 3 & 5 & 6 & 7 & 9 & 11 & 12 & 15 & 0 & 0 & 0 & 0 & 2 \\
$\mathcal{B}$  & 1 & 0 & 0 & 1 & 0 & 0 & 1 & 0 & 0 & 0 & 0 & 0 & 0 & 0 & 0 & 1 & 1 & 1 & 1 & 0 \\
$\mathcal{C}$  & 1 & 1 & 1 & 2 & 2 & 2 & 3 & 3 & 3 & 3 & 3 & 3 & 3 & 3 & 3 & 4 & 5 & 6 & 7 & 7 \\
};
\node[above=1em of a] (labela) {New index 9};
\node[right=1em of labela] (labelb) {Old index 9};
\draw [->,thick] (labela.south) -- (a.north);
\draw [->,thick] (labelb.south)-- (b.north east);
\draw (0,-2.8) node[text width=10.5cm] {\textbf{b)} The result of $\pdivide(8,3)$ on the structure of a). Representative value 30 was removed from $\mathcal{R}$. We shifted and updated $\mathcal{U}$, $\mathcal{B}$ and $\mathcal{C}$ to remove the old representative and accommodate for a new element with value 2.};
\end{tikzpicture}

\vskip 0.5cm

\begin{tikzpicture}
\matrix (first) [table1withindices]
{
  & 1 & 2 & 3 & 4 & 5 & 6 & 7 & 8 & 9 & 10 & 11 &\node (a) {12}; & 13 & 14 & 15 & 16 & 17 & 18 & 19 \\
  $Z$   & 5 & 1 & 4 & 7 & 1 & 1 & 6 & 3 & 2 & 1 & 1 & 4 & 1 & 3 & 5 & 10 & 5 & 10 & 2 \\
  $Y$   & \node[fill=black!20] () {5}; & 6 & 10 & \node[fill=black!20] () {17}; & 18 & 19 & \node[fill=black!20] () {25}; & 28 & 30 & 31 & 32 & 36 & 37 & 40 & \node[fill=black!20] () {45}; & \node[fill=black!20] () {55}; & \node[fill=black!20] () {60}; & \node[fill=black!20] () {70}; & 72 \\
$\mathcal{R}$  & $\{5, 17, 25, 45, 55, 60, 70\}$ \\
$\mathcal{U}$   & 0 & 1 & 5 & 0 & 1 & 2 & 0 & 3 & 5 & 6 & 7 & 11 & 12 & 15 & 0 & 0 & 0 & 0 & 2 \\
$\mathcal{B}$  & 1 & 0 & 0 & 1 & 0 & 0 & 1 & 0 & 0 & 0  & 0 & 0 & 0 & 0 & 1 & 1 & 1 & 1 & 0 \\
$\mathcal{C}$  & 1 & 1 & 1 & 2 & 2 & 2 & 3 & 3 & 3 & 3 & 3 & 3 & 3 & 3 & 4 & 5 & 6 & 7 & 7 \\
};
\node[above=1em of a] (labela) {Index containing the sum of the merged indices.};
\draw [->,thick] (labela.south) -- (a.north);
\draw (0,-2) node[text width=13cm, align=center] {\textbf{c)} The result of $\pmerge(12)$ on the structure of c).}; 
\end{tikzpicture}

\caption{Illustrating operations on the data structure with $B2^{\delta}=4$. a) shows the data structure immediately after a rebuild, b) shows the result of performing $\pdivide(8,3)$ on the structure of a), and c) shows the result of performing $\pmerge(12)$ on the structure of b).\label{fig:exampleinit}}
\end{figure}

%% file: substringconcat.tex
\section{Substring Concatenation}\label{app:substringconcat}
In this section we prove Theorem~\ref{thm:substringconcat}. Recall that we must store a string $R$ subject to substring concatenation queries: given two strings $x$ and $y$ return the location of an occurrence of $xy$ in $R$ or \texttt{NO} if no such occurrence exist.



To prove (i) we need the following definitions. For a substring $x$ of $R$, let $S(x)$ denote the suffixes of $R$ that have $x$ as a prefix, and let $S'(x) = \{i+|x| \mid i \in S(x) \wedge i+|x| \leq n\}$, i.e., $S'(x)$ are the suffixes of $R$ that are immediately preceded by $x$. Hence for two substrings $x$ and $y$, the suffixes that have $xy$ as a prefix are exactly $S'(x) \cap S(y)$. 
We can reduce this intersection problem to a 1D range emptiness problem as follows.

Let $\rank(i)$ be the position of suffix $R[i..r]$ in the lexicographic ordering of all suffixes of $R$, and let $\rank(A) = \{\rank(i) \mid i \in A\}$ for $A \subseteq \{1..n\}$. Then $xy$ is a substring of $R$ if and only if $\rank(S'(x)) \cap \rank(S(y)) \not = \emptyset$. Note that $\rank(S(y))$ is a range $[a,b] \subseteq [1,n]$, and we can determine this range in constant time for any substring $y$ using a constant-time weighted ancestor query on the suffix tree of $R$~\cite{gawrychowski2014weighted}. Consequently, we can decide if $xy$ is a substring of $R$ by a 1D range emptiness query on the set $\rank(S'(x))$. 

Belazzougui~et~al.~\cite{belazzougui2010fast} (see also \cite{goswami2015}) recently gave a 1D range emptiness data structure for a set $A \subseteq [1,r]$ using $O(|A|\log^\epsilon r)$ \emph{bits} of space, for any constant $\epsilon > 0$, and answering queries in constant time. We will build this data structure for $\rank(S'(x))$, but doing so for all substrings would require space $\tilde \Omega(r^2)$. 

To arrive at the space bound of $O(r\log^\epsilon r)$ (words), we employ a heavy path decomposition~\cite{harel1984fast} on the suffix tree of $R$, and only build the data structure for substrings of $R$ that correspond to the top of a heavy path. In this way, each suffix will appear in at most $\log r$ such data structures, leading to the claimed $O(r\log^\epsilon r)$ space bound (in words). In addition, we build a $O(r)$-space nearest common ancestor data structure~\cite{harel1984fast} for the suffix tree of $R$. Constant-time nearest common ancestor queries will allow us to also answer longest common prefix  queries on $R$ in constant time. 

To answer a substring concatenation query with substrings $x$ and $y$, we first determine how far $y$ follows the heavy path in the suffix tree from the location where $x$ stops. This can be done in $O(1)$ time by a constant-time longest common prefix query between two suffixes of $R$. We then proceeed to the top of the next heavy path, where we query the 1D range reporting data structure with the range $\rank(S(y'))$ where $y'$ is the remaining unmatched suffix of $y$. This completes the query, and the proof of (i).

The second solution (ii) is an implication of a result by Bille~et~al.~\cite{BGVV2014}. Given the suffix tree $ST_R$ of $R$, an \emph{unrooted longest common prefix query}~\cite{cole2004dictionary} takes a suffix $y$ and a location $\ell$ in $ST_R$ (either a node or a position on an edge) and returns the location in $ST_S$ that is reached after matching $y$ starting from location $\ell$. A substring concatenation query is straightforward to implement using two unrooted longest common prefix queries, the first one starting at the root, and the second starting from the location returned by the first query. It follows from Bille~et~al.~\cite{BGVV2014} that we can build a linear space data structure that supports unrooted longest common prefix queries in time $O(\log \log r)$ thus completing the proof of (ii).

%% file: extensions.tex
\section{Extensions}\label{sec:extensions}
In this section we show how to solve two other variants of the dynamic relative compression problem. We first prove Theorem~\ref{thm:extension1}, showing how to improve the query time if only supporting operations $\sacc$ and $\srep$. We then show Theorem~\ref{thm:extension2}, generalising the problem to support multiple strings. These data structures use the same substring concatenation data structure of Theorem~\ref{thm:substringconcat} as before but replaces the dynamic partial sums data structure.

\subsection{Dynamic Relative Compression with Access and Replace}
In this setting we constrain the operations on $S$ to $\sacc(i)$ and $\srep(i, \alpha)$.
Then, instead of maintaining a dynamic partial sums data structure over the lengths of the substrings in $C$, we only need a dynamic predecessor data structure over the prefix sums. 
The operations are implemented as before, except that for $\sacc(i)$ we obtain block $b_j$ by computing the predecessor of $i$ in the predecessor data structure, which also immediately gives us access to the local index in $b_j$. For $\srep(i, \alpha)$, a constant number of updates to the predecessor data structure is needed to reflect the changes. We use substring concatenation queries to restore maximality as described in Section~\ref{sec:DRC}. The prefix sums of the subsequent blocks in $C$ are preserved since $|b_j|=|b_j^1|+|b_j^2|+|b_j^3|$.

With a linear space implementation of the van Emde Boas data structure \cite{Boas1977, BKZ1977,MN1990} we can support the predecessor queries and updates in $O(\log \log N)$ expected time. For substring concatenation we apply Theorem~\ref{thm:substringconcat}(ii) using $O(r)$ space and $O(\log \log r)$. Since the length of source string does not change, we can always assume that $r > N$, and the total time becomes $O(\log \log N + \log \log r) = O(\log \log N)$. In summary, this proves Theorem~\ref{thm:extension1}.

\subsection{Dynamic Relative Compression with Split and Concatenate}
Consider the variant of the dynamic relative compression problem where we want to maintain a relative compression of a set of strings $S_1,\ldots,S_k$. Each string $S_i$ has a cover $C_i$ and all strings are compressed relative to the same string $R$. In this setting $n=\sum_{i=1}^k |C_i|$. In addition to the operations $\sacc$, $\srep$, $\sins$, and $\sdel$, we also want to support split and concatenation of strings. Note that the semantics of the operations change to indicate the string(s) to perform a given operation on. 

We build a leaf-oriented height-balanced binary tree $T_i$ (e.g. an AVL tree or red-black tree) over the blocks $C_i[1],\ldots,C_i[|C_i|]$ for each string $S_i$. In each internal node $v$, we store the sum of the block sizes represented by its leaves. Since the total number of blocks is $n$, the trees use $O(n)$ space. All operations rely on the standard procedures for searching, inserting, deleting, splitting and joining height-balanced binary trees. All of these run in $O(\log n)$ time for a tree of size $n$. See for example \cite{CLRS2001} for details on how red-black trees achieve this.

The answer to an $\sacc(i, j)$ query is found by doing a top-down search in $T_i$ using the sums of block sizes to navigate. Since the tree is balanced and the size of the cover is at most $n$, this takes $O(\log n)$ time. The operations $\srep(i,j,\alpha)$, $\sins(i,j,\alpha)$, and $\sdel(i,j)$ all initially require that we use $\sacc(i,j)$ to locate the block containing the $j$-th character of $S_i$. To reflect possible changes to the blocks of the cover, we need to modify the corresponding tree to contain more leaves and restore the balancing property. Since the number of nodes added to the tree is constant these operations each take $O(\log n)$ time. The $\sconcat(i,j)$ operation requires that we join two trees in the standard way and restore the balancing property of the resulting tree. For the $\ssplit(i,j)$ operation we first split the block that contains position $j$ such that the $j$-th character is the trailing character of a block. We then split the tree into two trees separated by the new block. This takes $O(\log n)$ time for a height-balanced tree.

To finalize the implementation of the operations, we must restore the maximality property of the affected covers as described in Section~\ref{sec:DRC}. At most a constant number of blocks are non-maximal as a result of any of the operations. If two blocks can be combined to one, we delete the leaf that represents the rightmost block, update the leftmost block to reflect the change, and restore the property that the tree is balanced. If the tree subsequently contains an internal node with only one child, we delete it and restore the balancing. Again, this takes $O(\log n)$ time for balanced trees, which concludes the proof of Theorem~\ref{thm:extension2}.

%% file: conclusion.tex
\section{Conclusion}

We have shown how to compress a text relatively to a reference string while supporting access to the text and a range of dynamic operations under some strong guarantees for the space usage and the query times. There are, however, room for improvement.

Our solution to DRC is built on data structures for the partial sums problem and the substring concatenation problem. Our partial sums-solution is optimal, but in order to get the desired constant query time for substring concatenation, our data structure uses $O(r\log^\epsilon r)$ space. As opposed to this, our linear space solution leads to $O(\log\log r)$ query time. We leave as an open problem if it is possible to get $O(1)$ time substring concatenation queries using $O(r)$ space, which will also carry over to a stronger result for the DRC problem.

Moreover, the size of the cover that is maintained by our DRC data structure is also an interesting parameter. Currently we maintain a 2-approximation of the optimal cover. It would be useful to know if a better approximation ratio can be maintained under the same (or better) time and space bounds that we give.